\documentclass[10pt]{article}
\usepackage{amsfonts,color}
\usepackage{amssymb,amsmath,amsthm,latexsym}
\usepackage{amsmath,amsthm,amssymb,arydshln}
\newtheorem{Theorem}{Theorem}[section]
\newtheorem{lem}[Theorem]{Lemma}

\newtheorem{Corollary}[Theorem]{Corollary}

\newtheorem{Example}[Theorem]{Example}
\numberwithin{equation}{section}

\begin{document}
\title{Constacyclic symbol-pair codes: lower bounds and optimal constructions\footnote{
 E-Mail addresses: bocong\_chen@yahoo.com (B. Chen), L\_R\_Lin86@163.com (L. Lin),  hwliu@mail.ccnu.edu.cn (H. Liu).}}

\author{Bocong Chen$^1$, ~Liren Lin$^2$,  ~Hongwei Liu$^2$}

\date{\small
${}^1$School of Mathematics, South China University of Technology, Guangzhou,
Guangdong, 510641, China\\
${}^2$School of Mathematics and Statistics,
Central China Normal University,
Wuhan, Hubei, 430079, China\\         }

\maketitle

\begin{abstract}
Symbol-pair codes introduced by  Cassuto and Blaum (2010) are designed  to
protect   against pair errors  in symbol-pair read channels.
The higher
the minimum pair distance, the more  pair errors the code can correct.
MDS symbol-pair codes are optimal in the sense that
 pair distance cannot be improved for given length and code size.
The  contribution of this paper is twofold.
First we
present three lower bounds for the minimum pair distance
of constacyclic codes, the first two of which generalize the previously known results
due to Cassuto and Blaum (2011) and Kai {\it et al.} (2015).
The third one exhibits  a lower bound for the minimum pair distance of repeated-root cyclic codes.
Second we
obtain new MDS symbol-pair codes  with minimum pair distance seven and eight through repeated-root cyclic codes.


\end{abstract}

\section{Introduction}
Let $\Sigma$ be a set of size $q$,
which we refer to
as an  {\it  alphabet } and whose elements are called {\it  symbols}. A {\it $q$-ary code} $\mathcal{C}$ of length $n$
over $\Sigma$ is a nonempty subset of $\Sigma^n$.
For any vector $\mathbf{a}=(a_0, a_1, \cdots, a_{n-1})\in\Sigma^n$,  the {\it symbol-pair read vector} of $\mathbf{a}$
is defined to be
$$
\pi(\mathbf{a})=\big[(a_0,a_1), (a_1,a_2), \cdots, (a_{n-2}, a_{n-1}), (a_{n-1}, a_0)\big].
$$
Two pairs $(c,d)$ and $(e,f)$ are distinct  if $c\neq e$ or $d\neq f$, or both.
The {\it pair distance} between $\mathbf{a}$ and $\mathbf{b}$, denoted by $d_p(\mathbf{a},\mathbf{b})$,
is defined as
$$
d_p\big(\mathbf{a},\mathbf{b}\big)=d_H\big(\pi(\mathbf{a}),\pi(\mathbf{b})\big),
$$
where $d_H$ denotes the usual Hamming distance.
It turns out that the set $\Sigma^n$ equipped with the pair distance $d_p$ is indeed a {\it metric space} (see \cite{Cassuto11}).
In a similar way to Hamming-metric codes, the {\it minimum pair distance} of a code $\mathcal{C}$ is defined to be
$$
d_p\big(\mathcal{C}\big)=\min\big\{d_p\big(\mathbf{a},\mathbf{b}\big)\,\big{|}\,\mathbf{a},\mathbf{b}\in \mathcal{C}, \mathbf{a}\neq \mathbf{b}\big\}.
$$
For any code $\mathcal{C}$ of length $n$ with $0<d_H(\mathcal{C})<n$, a simple but important
connection between $d_H(\mathcal{C})$ and $d_p(\mathcal{C})$ is  given in \cite{Cassuto11}:
$d_H(\mathcal{C})+1\leq d_p(\mathcal{C})\leq 2d_H(\mathcal{C})$.
A code of length $n$ over $\Sigma$
is called an {\it $(n,M,d_p)_q$-symbol-pair code} if its size is $M$ and minimum pair distance is $d_p$.

Symbol-pair codes introduced by  Cassuto and Blaum \cite{Cassuto10,Cassuto11}
are designed  to  protect   against pair errors  in symbol-pair read channels, where the outputs are overlapping pairs of symbols.
The   seminal works \cite{Cassuto10,Cassuto11,CL11} have
established   relationships between the minimum Hamming distance
of an error-correcting code and the minimum pair distance,
have found methods for  code constructions
and decoding,  and have obtained  lower and upper bounds on code sizes.
It was shown in \cite{Cassuto11} that if a code has minimum pair distance $d_p$ then it
can correct up to $\lfloor (d_p-1)/2\rfloor$ symbol-pair errors.
For this reason, it is desirable  to construct symbol-pair codes having a large minimum pair distance.

For a fixed code length $n$, it would certainly be nice if both  the code size $M$ (which is a measure of the
efficiency of the code) and the minimum pair distance $d_p$
could be as large as possible. However,
as in the Hamming-metric  case,
these two parameters are restricted each other for any
fixed length.
The Singleton-type Bound
for symbol-pair codes  relates the parameters $n$, $M$ and $d_p$
(see \cite[Theorem 2.1]{Chee}):
If $\mathcal{C}$ is an
$(n,M,d_p)_q$-symbol-pair code with $q\geq2$ and $2\leq d_p\leq n$, then
\begin{equation}\label{singleton}
M\leq q^{n-d_p+2}.
\end{equation}
A symbol-pair code for which equality holds in (\ref{singleton}) is said to be
{\it maximum distance separable} (MDS).
In this case, the code size $M$ is fully determined  by $n, d_p$ and $q$.
Following \cite{Chee}, we use   $(n,d_p)_q$ to denote an MDS symbol-pair code  of length
$n$ over $\Sigma$ with minimum pair distance $d_p$ and size $M=q^{n-d_p+2}$.
MDS symbol-pair codes are optimal in the sense that
 no code of length $n$
with $M$ codewords has a larger minimum pair distance than an MDS symbol-pair code with parameters $n$
and $M$. Constructing MDS symbol-pair codes is thus of
significance in theory and practice.

Cassuto and Blaum  \cite{Cassuto11} studied how the class of   cyclic codes can be exploited
as a framework for  symbol-pair  codes.
Combining  the  discrete Fourier transform (DFT) with the BCH Bound,
\cite[Theorem 10]{Cassuto11} showed that if  the generator polynomial of a  simple-root $[n,k,d_H]$ cyclic code
has at least $d_H$ roots (in some extension field over $\mathbb{F}_q$),
then the minimum pair distance of the code is at least $d_H+2$.
Using the Hartmann-Tzeng Bound, this lower bound was improved to $d_H+3$ when the code length $n$ is a prime number and
a constraint condition  on $n,k$ and $d_H$ is assumed (see \cite[Theorem 11]{Cassuto11}).
In a follow-up paper \cite{Kai},  Kai {\it et al.}     showed  that
\cite[Theorem 10]{Cassuto11} can be generalized to  simple-root constacyclic codes:
If  the generator polynomial of a  simple-root $[n,k,d_H]$ constacyclic code
has at least $d_H$ roots,
then the minimum pair distance of the code is at least $d_H+2$ (see \cite[Lemma 4.1]{Kai}).
Recently, Yaakobi {\it et al.}  \cite[Theorem 4]{Yaakobi}
obtained an elegant result on the minimum pair distance of binary   cyclic codes:
If  $\mathcal{C}$ is a binary   cyclic code of dimension greater than one, then
$
d_p(\mathcal{C})\geq d_H(\mathcal{C})+\lceil\frac{d_H(\mathcal{C})}{2}\rceil.
$

After establishing the  Singleton Bound (\ref{singleton}) for symbol-pair codes,
Chee  {\it et al.} \cite{Chee,Chee12} employed various methods to construct   MDS symbol-pair codes, including
the use of classical MDS codes, interleaving method of Cassuto and Blaum \cite{Cassuto11},  and
eulerian graphs of certain girth, etc.  It is worth noting that in contrast with all known classical MDS codes,
of which the lengths  are so small with respect to the alphabet
size,  MDS symbol-pair codes can have relatively large code length (see \cite{Chee}). In the light of
the Singleton Bound (\ref{singleton}) and
\cite[Lemma 4.1]{Kai},
Kai {\it et al.}  \cite{Kai} used almost MDS  constacyclic codes to construct MDS symbol-pair codes;
several classes of almost MDS constacyclic codes with minimum Hamming distance three or four
are constructed, and,  consequently,  MDS symbol-pair codes
with minimum pair distance five or six are obtained.

The aforementioned works lead us to the study of  lower bounds for the minimum pair distance of   constacyclic codes
and   constructions  of MDS symbol-pair codes.
The  contribution of this paper is twofold.
First we
present three lower bounds for the minimum pair distance
of constacyclic codes, the first two of which generalize the previously known results
\cite[Theorem 10]{Cassuto11}, \cite[Theorem 11]{Cassuto11} and \cite[Lemma 4.1]{Kai}.
The third one exhibits  a lower bound for the minimum pair distance of repeated-root cyclic codes.
Second we
construct new MDS symbol-pair codes  with minimum pair distance seven and eight by using repeated-root cyclic codes.
More precisely,
we summarize our results as follows.

Thereafter, $\mathbb{F}_q$ denotes a finite field of size $q$, where $q$ is a power of a prime number $p$.
Let $n>1$ be a positive integer ($n$ and $p$ are not necessarily co-prime).
\begin{Theorem}\label{add2}
Let $\mathcal{C}$ be an $[n,k,d_H]$  constacyclic code over  $\mathbb{F}_q$ with $2\leq d_H<n$. Then  we have  the following.
\begin{itemize}
\item[$(1)$]
$d_p(\mathcal{C})\geq d_H+2$ if and only if $\mathcal{C}$ is not an MDS code, i.e.,  $k<n-d_H+1$.
Equivalently, $d_p(\mathcal{C})=d_H+1$  if and only if $\mathcal{C}$ is   an MDS code, i.e.,  $k=n-d_H+1$.

\item[$(2)$]
If $k>1$ and $n-d_H\geq 2k-1$, then $d_p(\mathcal{C})\geq d_H+3$.

\end{itemize}
\end{Theorem}

\begin{Theorem}\label{re-add3}
Let $\mathcal{D}$ be a nonzero $[\ell p^e, k, d_H]$ repeated-root cyclic code  over $\mathbb{F}_q$
with generator polynomial $g(x)$, where $\ell>1$ is a positive integer
co-prime to $p$ and $e$ is a positive integer.
If $d_H(\mathcal{D})$ is a prime number and if one of the following two conditions is satisfied
\begin{itemize}
\item[$(1)$]
 $\ell<d_H(\mathcal{D})<\ell p^e-k$;
 \item[$(2)$]
  $x^\ell-1$ is a divisor of $g(x)$ and $2<d_H(\mathcal{D})<\ell p^e-k$,
\end{itemize}
then
$d_p(\mathcal{D})\geq d_H(\mathcal{D})+3$.
\end{Theorem}
At this point we make several remarks. The first part of
Theorem \ref{add2} extends \cite[Theorem 10]{Cassuto11} and \cite[Lemma 4.1]{Kai} in two directions: First we improve the results
by giving a necessary and sufficient condition. Second we do not require that
$\gcd(n,q)=1$.

We make a comparison between \cite[Theorem 11]{Cassuto11} and the second part of Theorem \ref{add2}.
\cite[Theorem 11]{Cassuto11} says that
if a $q$-ary $[n,k,d_H]$ simple-root  cyclic code with prime length $n$ satisfies $n-d_H\geq 2k-2$,
then  the minimum pair distance of the code is at least $d_H+3$.
The second part of Theorem \ref{add2}   removes the prime-length  constraint and the simple-root requirement;
if $n-d_H$ is odd, the conditions $n-d_H\geq 2k-1$ and $n-d_H\geq 2k-2$ coincide; otherwise,
the two conditions are equivalent to $k\leq (n-d_H)/2$ and $k\leq (n-d_H)/2+1$ respectively.

Using Theorems \ref{add2} and \ref{re-add3}, we obtain the following new MDS symbol-pair codes.

\begin{Theorem}\label{MDSs}
The following hold.
\begin{itemize}

\item[$(1)$]Let $p\geq5$ be an odd prime number.
Then there exists an MDS $(3p, 7)_p$-symbol-pair code.

\item[$(2)$]Let $p$ be an odd prime number such that $3$ is a divisor of $p-1$.
Then there exists an MDS $(3p, 8)_p$-symbol-pair code.

\item[$(3)$]Let $p\geq5$ be an odd prime number.
Then there exists an MDS $(3p, 6)_p$-symbol-pair code.

\item[$(4)$]Let $q\geq3$ be a  prime power and let $n\geq q+4$ be a divisor of $q^2-1$. Then
there exists an MDS $\big(n,6\big)_q$-symbol-pair code.
\end{itemize}
\end{Theorem}
Note that \cite[Theorem 4.3]{Kai} asserts that there exists  an MDS $\big(n,5\big)_q$-symbol-pair code
if $n>q+1$ is a divisor of $q^2-1$. The fourth  part of Theorem \ref{MDSs} shows that the minimum pair distance $5$ can be increased to $6$.

This paper is organized as follows. Basic
notations and results about  constacyclic
codes and repeated-root cyclic codes are provided in Section \ref{preliminaries}.  The proofs of  Theorems \ref{add2} and \ref{re-add3},
together with some corollaries and examples,  are presented in
Section \ref{proof-bound}. The proof of Theorem \ref{MDSs} is given in Section \ref{proof-MDS}.

\section{Preliminaries}\label{preliminaries}
In this section,    basic notations and results
about  constacyclic codes and repeated-root cyclic codes  are provided.
The result
\cite[Theorem 1]{Cast}  plays an important role in the proof of Theorems \ref{re-add3} and \ref{MDSs},
which   provides an effective way to determine  the minimum
Hamming distance of repeated-root cyclic codes.

A {\it   code} $\mathcal{C}$ of length $n$ over $\mathbb{F}_q$ is a nonempty subset of $\mathbb{F}_q^n$.
If, in addition,  $\mathcal{C}$  is a linear subspace over $\mathbb{F}_q$ of $\mathbb{F}_q^n$,
then $\mathcal{C}$ is called a {\it  linear code}.
A linear code $\mathcal{C}$ of length $n$,  dimension $k$ and minimum Hamming distance $d_H$ over $\mathbb{F}_q$ is often
called a $q$-ary $[n, k, d_H]$ code.
Given  a nonzero element $\lambda$ of $\mathbb{F}_q$,
the $\lambda$-constacyclic   shift $\tau_{\lambda}$ on $\mathbb{F}_q^n$ is the shift
$$
\tau_{\lambda}\big((x_0,x_1,\dots,x_{n-1})\big)=\big(\lambda x_{n-1},x_0,x_1,\dots,x_{n-2}\big).
$$
A linear code $\mathcal{C}$ is said to be {\it $\lambda$-constacyclic} if $\mathcal{C}$ is a $\tau_{\lambda}$-invariant subspace of $\mathbb{F}_q^n$, i.e.,   $\tau_{\lambda}(\mathcal{C})=\mathcal{C}$.
In particular, it is just
the usual {\it cyclic code} when $\lambda=1$.
In studying constacyclic codes of length $n$, it is convenient to label the coordinate positions as
$0,1,\cdots, n-1$. Since a
constacyclic code of length $n$ contains all $n$ constacyclic shifts of any codeword,
it is convenient to think
of the coordinate positions cyclically where, once you reach $n-1$, you begin again with
coordinate $0$. When we speak of consecutive coordinates, we will always mean consecutive
in that cyclical sense.

Each codeword $\mathbf{c}=(c_0,c_1,\dots,c_{n-1})\in \mathcal{C}$ is customarily identified with its polynomial representation $c(x)=c_0+c_1x+\dots+c_{n-1}x^{n-1}$.
Any  code $\mathcal{C}$ is then in turn identified with the set of all polynomial representations of its codewords.
In this way,
a linear code $\mathcal{C}$ is $\lambda$-constacyclic if and only if it is an ideal of the
quotient ring $\mathbb{F}_q[x]/\langle x^n-\lambda\rangle$ (e.g.,  see \cite{CDFL}).
It follows that  a unique  monic divisor   $g(x)\in \mathbb{F}_q[x]$ of $x^n-\lambda$ can be found  such that
$\mathcal{C}=\langle g(x)\rangle=\big\{f(x)g(x)\pmod{x^n-\lambda}\,\big{|}\,f(x)\in \mathbb{F}_q[x]\big\}$.
The polynomial $g(x)$ is called the {\it generator polynomial} of $\mathcal{C}$,
in which  case  $\mathcal{C}$ has dimension $k$ precisely when the degree of $g(x)$ is $n-k$.

Generally, constacyclic codes over finite fields can be divided into two classes:
simple-root constacyclic codes, if the code lengths are co-prime to the characteristic of the field; otherwise, we have the so-called
repeated-root constacyclic codes.
Most of studies on constacyclic codes in the literature are focused on the simple-root case, which
essentially guarantees that every root of $x^n-\lambda$ has multiplicity one.
Simple-root constacyclic codes are thus can be characterized by their defining sets (e.g., see \cite{Chen} or \cite{KZL}).
The BCH Bound and the Hartmann-Tzeng Bound  for simple-root
cyclic codes (e.g., see \cite{Huffman}) are based on consecutive sequences of roots of the generator polynomial.

In contrast to the simple-root case, repeated-root constacyclic codes  are no longer characterized by sets
of zeros.
Castagnoli {\it et al.} \cite[Theorem 1]{Cast} determined  the minimum Hamming distance  of repeated-root cyclic codes
by using polynomial algebra;
it is  showed that the minimum Hamming distance of a repeated-root cyclic code
$\mathcal{D}$ can be expressed in terms of
$d_H(\bar{\mathcal{D}_t})$, where $\bar{\mathcal{D}_t}$ are   simple-root cyclic codes fully  determined by $\mathcal{D}$.
To include  \cite[Theorem 1]{Cast}, we first introduce  the following notation.
Let $\mathcal{D}=\langle g(x)\rangle$ be a   repeated-root cyclic code of length $\ell p^e$ over $\mathbb{F}_q$,
where $\ell>1$ is a positive integer such that $\gcd(\ell,p)=1$ and $e$ is a positive integer. Suppose
$$
g(x)=\prod\limits_{i=1}^sm_i(x)^{e_i}
$$
is the factorization of $g(x)$ into distinct monic irreducible polynomials $m_i(x)\in \mathbb{F}_q[x]$ of multiplicity $e_i$.
Fix a value $t$, $0\leq t\leq p^e-1$;  $\bar{\mathcal{D}}_t$ is defined to be a (simple-root) cyclic code of length
$\ell$ over $\mathbb{F}_q$ with generator polynomial
$g_t(x)$ as the product of those irreducible factors $m_i(x)$ of $g(x)$ that occur with multiplicity $e_i>t$.
If this product turns out to be $x^\ell -1$, then
$\bar{\mathcal{D}_t}$ contains only the all-zero codeword and we set $d_H(\bar{\mathcal{D}_t})=\infty$.
If all $e_i ~(1\leq i\leq s)$ satisfy $e_i\leq t$, then, by way of convention, $g_t(x)=1$ and $d_H(\bar{\mathcal{D}_t})=1$.
The next result is an immediate consequence of  \cite[Lemma 1]{Cast} and \cite[Theorem 1]{Cast}.

\begin{lem}\label{repeated-root}
Let $\mathcal{D}=\langle g(x)\rangle$ be a   repeated-root cyclic code of length $\ell p^e$ over $\mathbb{F}_q$, where $\ell>1$ is a positive integer such that $\gcd(\ell,p)=1$ and $e$ is a positive integer.
Then
$$
d_H(\mathcal{D})=\min\big\{P_t\cdot d_H(\bar{\mathcal{D}}_t)\,\big{|}\, 0\leq t\leq p^e-1\big\}
$$
where
\begin{equation}\label{pt}
P_t=\prod_i\big(t_i+1\big)
\end{equation}
with $t_i$'s being the coefficients of the radix-$p$ expansion
of $t$.
\end{lem}

\section{Proofs of Theorems \ref{add2} and \ref{re-add3}}\label{proof-bound}
The proof of Theorem \ref{add2} is given below.\\
\vspace{0.01cm}\\
{\bf Proof of Theorem \ref{add2}.}
To prove  $(1)$, we first observe that
$d_p(\mathcal{C})\geq d_H+1$ since the minimum Hamming distance of $\mathcal{C}$  satisfies  $2\leq d_H<n$.
We will show that
$d_p(\mathcal{C})=d_H+1$ if and only if $k=n-d_H+1$.
To this end,
we claim that $d_p(\mathcal{C})=d_H+1$ precisely when $\mathcal{C}$ has a  codeword with Hamming weight  $d_H$  in the form
$$
\big(a_{i_1}, a_{i_2},\cdots, a_{i_d}, 0, \cdots, 0\big),
$$
 where $a_{i_j}$ are nonzero elements of $\mathbb{F}_q$
for $1\leq j\leq d$ (here, $d_H$ is denoted by $d$ for short). Indeed, it is clear that
$d_p(\mathcal{C})=d_H+1$ if and only if  there exists a  codeword  $\mathbf{c}\in \mathcal{C}$
such that $w_H(\mathbf{c})=d_H$ and the $d_H$ nonzero terms   appear with consecutive coordinates;
applying the $\lambda$-constacyclic shift a certain number of times on $\mathbf{c}$
if necessary, $\mathbf{c}$ is then converted to the form
$(a_{i_1}, \cdots, a_{i_d}, 0, \cdots, 0)$, as claimed.

If $\mathcal{C}$ is an MDS code, namely $k=n-d_H+1$,  then $d_p(\mathcal{C})=d_H+1$.
For the converse, let $H=(\mathbf{h}_1, \cdots, \mathbf{h}_n)$ be a parity-check matrix  for $\mathcal{C}$,
where $\mathbf{h}_i$ ($1\leq i\leq n$) are the columns of $H$. Suppose $d_p(\mathcal{C})=d+1$, then there exists a codeword
$\mathbf{c}=(a_{i_1}, \cdots, a_{i_d}, 0, \cdots, 0)\in \mathcal{C}$, as claimed in the   preceding paragraph.
Hence,  $a_{i_1}\mathbf{h}_1+\cdots+a_{i_d}\mathbf{h}_d=0$, which implies that
the $d$th column $\mathbf{h}_d$ lies in the $(d-1)$-dimensional subspace of $\mathbb{F}_q^{n-k}$ spanned by $\mathbf{h}_1, \mathbf{h}_2, \cdots, \mathbf{h}_{d-1}$, say
$V=\langle \mathbf{h}_1, \cdots, \mathbf{h}_{d-1}\rangle$. Using the $\lambda$-constacyclic shift on $\mathbf{c}$, it follows that
$(0,a_{i_1}, \cdots, a_{i_d}, 0, \cdots, 0)$ is also a codeword of $\mathcal{C}$.
Therefore, $a_{i_1}\mathbf{h}_2+\cdots+a_{i_d}\mathbf{h}_{d+1}=0$. This leads to $\mathbf{h}_{d+1}\in V$.
We can continue in
this fashion and eventually  obtain that the dimension of
the vector space generated by the columns of $H$ is exactly equal to $d_H-1$.
However, $H$ is a full row-rank matrix of size $(n-k)\times n$, which forces $n-k=d_H-1$.
This completes the proof of  $(1)$.

The proof of   Theorem \ref{add2}($2$) needs the following corollary.
Using essentially identical arguments to the proof Theorem \ref{add2}$(1)$, we have the following result.

\begin{Corollary}\label{corollary}
Let $\mathcal{C}$ be an $[n,k,d_H]$  constacyclic code over  $\mathbb{F}_q$ with $2\leq d_H<n$.
If $\mathcal{C}$ contains a codeword, of which the Hamming weight is  $d_H+1$,  such that
the $d_H+1$ nonzero terms   appear with consecutive coordinates,
then   $n-d_H\leq k$.
\end{Corollary}

Now we continue to  give the proof of   Theorem \ref{add2}$(2)$.
Since the parameters of $\mathcal{C}$ satisfy $n-d_H\geq 2k-1$,  it follows from   Theorem \ref{add2}$(1)$  that
$d_p(\mathcal{C})\geq d_H+2$.
In order to prove  $d_p(\mathcal{C})\geq d_H+3$, it suffices to show that
there are no codewords of $\mathcal{C}$ with Hamming weight $d_H+1$ such that
the $d_H+1$ nonzero terms   appear with consecutive coordinates,
and that there are no codewords of $\mathcal{C}$ with Hamming weight $d_H$
in the form
$(\mathbf{a},  \mathbf{0}_r,\mathbf{b},\mathbf{0}_s)$,
where $\mathbf{a},\mathbf{ b}$ are row vectors
with all the  entries of $\mathbf{a},\mathbf{ b}$ being nonzero, $\mathbf{0}_r$ and $\mathbf{0}_s$ are all-zero row vectors of lengths $r$ and $s$ respectively.

From  $k>1$ and $n-d_H\geq 2k-1$, we see that  $n-d_H>k$.
Using   Corollary \ref{corollary},  we are left to show that  there are no codewords of $\mathcal{C}$ with Hamming weight $d_H$
in the form
$(\mathbf{a},  \mathbf{0}_r,\mathbf{b},\mathbf{0}_s)$.
Suppose otherwise that  $\mathbf{c}=(\mathbf{a},  \mathbf{0}_r,\mathbf{b},\mathbf{0}_s)\in \mathcal{C}$
with $w_H(\mathbf{c})=d_H$. We will derive a contradiction.
Let $g(x)$ be the generator polynomial of $\mathcal{C}$. Then there exists a unique polynomial $u(x)$ with $\deg u(x)\leq k-1$ such that
$u(x)g(x)=c(x)=(\mathbf{a},  \mathbf{0}_r,\mathbf{b},\mathbf{0}_s)$.
If $s\geq k$, then  the degree of $c(x)$ is at most $n-k-1$.
This is impossible because the degree of $g(x)$ is $n-k$.
We thus conclude that $s\leq k-1$. Similar reasoning yields $r\leq k-1$.
This gives $n-d_H=r+s\leq 2k-2$, which contradicts
the hypotheses of the theorem. We are done.
\qed

\vspace{0.03cm}
We illustrate Theorem \ref{add2} in the following example.
\begin{Example}\rm
Take $q=5$ and $n=24$ in Theorem \ref{add2}. Let $\mathcal{C}$ be a cyclic code of length $24$ over
$\mathbb{F}_5$ with defining set $T=\mathbb{Z}_{24}\setminus\{0,19,23\}$.
\textsc{Magma} \cite{Magma} computations show that $\mathcal{C}$ has parameters $[24,3,19]$.
Since $24-19=5=2\times3-1$, it follows from  Theorem \ref{add2}$(2)$ that $d_p(\mathcal{C})\geq19+3=22$.
In fact, $\mathcal{C}$ has minimum pair distance $23$, which gives that $\mathcal{C}$ is an MDS
$(24,23)_5$-symbol-pair code.
\end{Example}

We now turn to the proof of Theorem \ref{re-add3}.\\
\vspace{0.01cm}\\
{\bf Proof of Theorem \ref{re-add3}.}
It follows from Theorem \ref{add2} that
$d_p(\mathcal{D})\geq d_H(\mathcal{D})+2$.
By $d_H(\mathcal{D})<\ell p^e-k$ again,
Corollary \ref{corollary}  ensures that
there are no  codewords of $\mathcal{D}$ with Hamming weight $d_H+1$ such that
the $d_H+1$ nonzero terms   appear with consecutive coordinates.
Therefore, it remains to show that
there are no codewords of $\mathcal{D}$ with  Hamming weight   $d_H$
in the form
\begin{equation}\label{forms}
\big(\mathbf{a},  \mathbf{0}_r,\mathbf{b},\mathbf{0}_s\big)
\end{equation}
where $\mathbf{a},\mathbf{ b}$ are row vectors
with all the  entries of $\mathbf{a},\mathbf{ b}$ being nonzero, $\mathbf{0}_r$ and $\mathbf{0}_s$ are all-zero row vectors of lengths $r$ and $s$ respectively.

To this end,
we first analyze the nonzero codewords of $\mathcal{D}$ by using \cite[Lemma 2]{Cast}.
Let $c(x)\in \mathcal{D}$   be an arbitrary nonzero codeword of degree at most $\ell p^e-1$.
Write $c(x)$ as $c(x)=(x^\ell-1)^tv(x)$, where $0\leq t\leq p^e-1$ and $x^\ell-1$ is not a divisor of $v(x)$,
and write $v(x)$ in the form
\begin{equation}\label{Nvv}
v(x)=v_0(x^\ell)+xv_1(x^\ell)+\cdots+x^{\ell-1}v_{\ell-1}(x^\ell).
\end{equation}
\cite[Lemma 2]{Cast} says that
$$
c_{\bar t}(x)=(x^\ell-1)^{\bar t}\bar{v}(x)^{p^e}\pmod{x^{\ell p^e}-1},
$$
where $\bar{v}(x)\equiv v(x)\pmod{x^\ell-1}$ and $\bar{t}=\min\{\bar {t}\in T\,|\, \bar{t}\geq t\}$
(The elements of $T$ are nonnegative integers; for the definition of $T$, the reader may refer to \cite{Cast}),
is also a nonzero codeword of $\mathcal{D}$  satisfying  $w_H(c_{\bar t}(x))\leq w_H(c(x))$.

Now choosing  $c(x)$ to be any codeword of $\mathcal{D}$ with Hamming weight $d_H$,
\cite[Lemma 2]{Cast} and \cite[Theorem 1]{Cast} together with their proofs  tell us more:
$$
d_H=w_H\big(c(x)\big)=w_H\big(c_{\bar t}(x)\big)=P_{\bar t}\cdot N_v,
$$
where $N_v$ is the number of nonzero $v_i(x^\ell)$'s in (\ref{Nvv})
and $P_{\bar t}$ is a positive integer defined in (\ref{pt}).
These facts yield
$d_H=P_{\bar t}$ or  $d_H=N_v$, with our assumption that $d_H$ is a prime number.
If $(1)$ holds,   we have   $N_v=1$ since $N_v\leq \ell$; if  $(2)$ holds,
it follows from $t\geq1$ that $\bar t\geq1$, and thus $P_{\bar t}\geq2$ which forces $N_v=1$.  In conclusion,
$c(x)$ must be one of the following forms:
\begin{equation*}
c(x)=x^i(x^\ell-1)^tv_{i}(x^\ell)~~\hbox{for~some~$0\leq i\leq \ell-1$}.
\end{equation*}
Expanding $c(x)$ and using the fact that the degree of $c(x)$ is  at most $\ell p^e-1$, it follows from
$d_H(\mathcal{D})\geq3$
that $c(x)$ cannot have the form (\ref{forms}).
This completes the proof.
\qed

We give two examples to illustrate Theorem \ref{re-add3}.

\begin{Example}\label{example2}\rm
Take $\ell=3$, $p=5$ and $e=1$ in Theorem \ref{re-add3}.
Let $\mathcal{D}$ be a repeated-root cyclic code of length $15$ over $\mathbb{F}_5$ with
generator polynomial $(x-1)(x^3-1)$. By Lemma \ref{repeated-root}, we see that $\mathcal{D}$
has parameters $[15,11,3]$.
It is readily checked that the conditions of Theorem \ref{re-add3} are satisfied, and thus  the minimum pair distance of $\mathcal{D}$
is at least $6$. Now the Singleton Bound for symbol-pair codes (\ref{singleton}) gives that $\mathcal{D}$ is an MDS $(15,6)_5$-symbol-pair code.
\end{Example}

Example \ref{example2} suggests an infinite  family of MDS symbol-pair codes with minimum pair distance six  as we show below.

\begin{Corollary}\label{cor2}
Let $p\geq5$ be an odd prime number.
Then there exists an MDS $(3p, 6)_p$-symbol-pair code.
\end{Corollary}
\begin{proof}
Let $\mathcal{D}$ be a repeated-root cyclic code of length $3p$ over $\mathbb{F}_p$ with generator polynomial
$(x-1)(x^3-1)$. Using Lemma \ref{repeated-root},  we see that  $d_H(\mathcal{D})=3$, and so $\mathcal{D}$  has parameters
$[3p,3p-4,3]$. Now the desired result follows from Theorem \ref{re-add3}.
\end{proof}

\begin{Example}\rm
Take $\ell=3$, $p=7$ and $e=1$ in Theorem \ref{re-add3}.
Let $\mathcal{D}$ be a repeated-root cyclic code of length $21$ over $\mathbb{F}_7$ with
generator polynomial $(x-1)^4(x-2)^2(x-4)$. Using Lemma \ref{repeated-root}, it is easy to see that $\mathcal{D}$
has parameters $[21,14,5]$.
The conditions of Theorem \ref{re-add3} are satisfied, and thus  the pair distance of $\mathcal{D}$
is at least $8$.   \textsc{Magma} \cite{Magma} computations show that $(6,4,1,1,\mathbf{0}_6,3,6,\mathbf{0}_9)$, where $\mathbf{0}_6$ and $\mathbf{0}_9$
denote   respectively  all-zero row vectors of length $6$ and $9$, is a codeword of $\mathcal{D}$.
Therefore, the true minimum pair distance of $\mathcal{D}$ is $8$.
\end{Example}

\section{Proof of Theorem \ref{MDSs}}\label{proof-MDS}
The proof of Theorem \ref{MDSs} is presented as follows.\\
\vspace{0.01cm}\\
{\bf Proof of Theorem \ref{MDSs}.}
$(1).$
Let $\mathcal{D}$ be a  cyclic code of length $3p$ over $\mathbb{F}_p$ with generator polynomial
$g(x)=(x-1)^3(x^2+x+1)$.
Using Lemma \ref{repeated-root},  we have that
$\mathcal{D}$ is a cyclic code over $\mathbb{F}_p$ with parameters $[3p, 3p-5, 4]$.
Theorem \ref{add2} gives $d_p(\mathcal{D})\geq6$ and Corollary \ref{corollary} implies that
there are no codewords of $\mathcal{D}$ with Hamming weight $5$ such that
the $5$ nonzero terms   appear with consecutive coordinates.
We are left to show that
there are no codewords of $\mathcal{D}$ with  Hamming weight   $4$
in the form
\begin{equation}\label{formform}
\big(\mathbf{a},  \mathbf{0}_u,\mathbf{b},\mathbf{0}_v\big)
\end{equation}
where $\mathbf{a},\mathbf{ b}$ are row vectors
with all the  entries of $\mathbf{a},\mathbf{ b}$ being nonzero, $\mathbf{0}_u$ and $\mathbf{0}_v$ are all-zero row vectors of lengths $u$ and $v$ respectively.
Let $c(x)$   be a minimum Hamming  weight  codeword of $\mathcal{D}$ of degree at most $3p-1$.
Write $c(x)$ as $c(x)=(x^3-1)^tv(x)$ where $0\leq t\leq p-1$ and $x^3-1$ is not a divisor of $v(x)$,
and write $v(x)$ in the form
\begin{equation}\label{Nvvv}
v(x)=v_0(x^3)+xv_1(x^3)+x^{2}v_{2}(x^3).
\end{equation}
Since $x^3-1$ is a divisor of the generator polynomial $g(x)$, we have $t\geq1$.
As pointed out in the proof of Theorem \ref{re-add3}, the following equalities  hold:
\begin{equation}\label{need}
4=w_H\big(c(x)\big)=(1+t)\cdot N_v,
\end{equation}
where $N_v$ is the number of nonzero $v_i(x^3)$'s in (\ref{Nvvv}).
There are two possible values for $N_v$: If
$t=3$, then $N_v=1$; if $t=1$, then $N_v=2$.
The case $N_v=1$ clearly implies that $c(x)$ cannot be in the form (\ref{formform}). Thus we only need to
consider the case $N_v=2$ and $t=1$.
Assume to the contrary that $c(x)=(x^3-1)v(x)$ is a minimum Hamming weight codeword of $\mathcal{D}$
in the form (\ref{formform}).
Without loss of generality we may suppose that the first coordinate of $c(x)$ is $1$.
There are two cases:

Case 1: $v(x)=v_0(x^3)+xv_1(x^3)$. The forms of  $(x^3-1)v_0(x^3)$ and $x(x^3-1)v_1(x^3)$ can be illustrated by the following table:
$$\begin{array}{c|ccc:ccc:ccc:c:cccc}
\hline
(x^3-1)v_0(x^3)&  1 & 0 & 0 & \Box & 0 & 0 & \Box & 0 & 0 &\cdots& \Box & 0 & 0\\
x(x^3-1)v_1(x^3)&  0 & \Box& 0 & 0 & \Box & 0 & 0 & \Box & 0 &\cdots&0 & \Box & 0\\
\hline
\end{array}
$$
where the symbol $\Box$ marks the possible nonzero terms.
To ensure that $c(x)$ is in the form (\ref{formform}), the Hamming weight of $(x^3-1)v_0(x^3)$ must be equal to $2$
and the coefficient of $x$ in the expansion of $x(x^3-1)v_1(x)$ must be nonzero.
Therefore, a positive integer $r$ with $1\leq r\leq p-1$ and three nonzero elements $a_1, a_2, a_3$ of $\mathbb{F}_p$ can be found such that
$$
c(x)=1+a_1x+a_2x^{3r}+a_3x^{3r+1}.
$$
With $c(1)=c(\omega)=c(\omega^2)=0$, we have $a_3=-a_1$ and $a_2=-1$.
On the other hand,  the first and the second formal derivative of $c(x)$ respectively gives
$$c^{(1)}(x)=a_1-3rx^{3r-1}-(3r+1)a_1x^{3r}$$
and
$$c^{(2)}(x)=-3r(3r-1)x^{3r-2}-3r(3r+1)a_1x^{3r-1}.$$
Since $(x-1)^3$ is a divisor of $c(x)$, it follows from $c^{(1)}(1)=c^{(2)}(1)=0$ that
$a_1=-1$ and $6r=0$. This is a contradiction, for $p\geq5$ is an odd prime number and $1\leq r\leq p-1$.

Case 2: $v(x)=v_0(x^3)+x^2v_2(x^3)$.
As in the previous case,
a positive integer $r$ with $1\leq r\leq p-1$ and three nonzero elements $a_1, a_2, a_3$ of $\mathbb{F}_p$ can be found such that
$$
c(x)=1+a_1x^{3r-1}+a_2x^{3r}+a_3x^{3p-1}.
$$
With  arguments similar to  the previous case,   we have   $6r=0$, a contradiction again.
This completes the proof of the first part of Theorem \ref{MDSs}.

$(2).$
Let $\mathcal{D}$ be a repeated-root cyclic code of length $3p$ over $\mathbb{F}_p$
with generator polynomial $(x-1)^3(x-\omega)^2(x-\omega^2)$.
Using Lemma \ref{repeated-root},  we have that
$\mathcal{D}$ is a cyclic code over $\mathbb{F}_p$ with parameters $[3p, 3p-6, 4]$.
Theorem \ref{add2} gives $d_p(\mathcal{D})\geq6$. Using techniques similar to
those used in the proof of Theorem \ref{add2}, we see that
there are no codewords of $\mathcal{D}$ with Hamming weight $5$ (resp. $6$) such that
the $5$ (resp. $6$) nonzero terms   appear with consecutive coordinates.

The proof will be completed in three steps.

Step 1.
There are no codewords of $\mathcal{D}$ with  Hamming weight   $4$
in the form
\begin{equation*}
\big(\mathbf{a},  \mathbf{0}_u,\mathbf{b},\mathbf{0}_v\big)
\end{equation*}
where $\mathbf{a},\mathbf{ b}$ are row vectors
with all the  entries of $\mathbf{a},\mathbf{ b}$ being nonzero, $\mathbf{0}_u$ and $\mathbf{0}_v$ are all-zero row vectors of lengths $u$ and $v$ respectively. It is trivial to see that this holds by  the same arguments as in the proof of $(1)$.

Step 2.
There are no codewords of $\mathcal{D}$ with  Hamming weight   $4$
in the form
\begin{equation}\label{formformform}
\big(\mathbf{a},  \mathbf{0}_u,\mathbf{b},\mathbf{0}_v,\mathbf{c},\mathbf{0}_w\big).
\end{equation}
where $\mathbf{a},\mathbf{ b}$ and $\mathbf{c}$ are row vectors
with all the  entries of $\mathbf{a},\mathbf{ b}$ and $\mathbf{c}$ being nonzero, $\mathbf{0}_u$,    $\mathbf{0}_v$ and $\mathbf{0}_w$ are all-zero row vectors of lengths $u$,  $v$ and $w$ respectively.
Assume to the contrary that $c(x)=(x^3-1)v(x)$ is a minimum Hamming weight codeword of $\mathcal{D}$
in the form (\ref{formformform}).
Without loss of generality we may suppose that the first coordinate of $c(x)$ is $1$.
At this point, we arrive at (\ref{need}) again.
There are two possible values for $N_v$: If
$t=3$, then $N_v=1$; if $t=1$, then $N_v=2$.
Clearly,  $c(x)$ cannot be in the form (\ref{formformform}) if  $N_v=1$.
We are left to consider the case $N_v=2$ and $t=1$. We now consider three cases separately.

Case 1. $c(x)=1+a_1x+a_2x^{3r}+a_3x^{3s}$, where   $r,s$
are positive integers  with   $1\leq r\neq s\leq p-1$ and  $a_1, a_2, a_3$ are nonzero elements of $\mathbb{F}_p$.
With $c(1)=c(\omega)=0$, we have
$$
1+a_1+a_2+a_3=0~~\hbox{and}~~1+a_1\omega+a_2+a_3=0,
$$
which forces $a_1=0$, a contradiction.

Case 2. $c(x)=1+a_1x+a_2x^{3r+1}+a_3x^{3s+1}$, where   $r,s$
are positive integers  with   $1\leq r\neq s\leq p-1$ and  $a_1, a_2, a_3$ are nonzero elements of $\mathbb{F}_p$.
It follows from $c(1)=c(\omega)$ that
$$
1+a_1+a_2+a_3=0
$$
and
$$
1+a_1\omega+a_2\omega+a_3\omega=0.
$$
This is impossible.

Case 3. $c(x)=1+a_1x+a_2x^{3r}+a_3x^{3s+1}$, where   $r,s$
are positive integers  with  $1\leq r\neq s\leq p-1$ and  $a_1, a_2, a_3$ are nonzero elements of $\mathbb{F}_p$.
With  $c^{(1)}(1)=c^{(1)}(\omega)=0$, we have
$$
a_1+3ra_2+(3s+1)a_3=0
$$
and
$$
a_1+3r\omega^2a_2+(3s+1)a_3=0,
$$
a contradiction.

Step 3.
There are no codewords of $\mathcal{D}$ with  Hamming weight   $5$
in the form
\begin{equation*}
\big(\mathbf{a},  \mathbf{0}_u,\mathbf{b},\mathbf{0}_v\big)
\end{equation*}
where $\mathbf{a},\mathbf{ b}$ are row vectors
with all the  entries of $\mathbf{a},\mathbf{ b}$ being nonzero, $\mathbf{0}_u$ and $\mathbf{0}_v$ are all-zero row vectors of lengths $u$ and $v$ respectively. It is easy to see that this case holds.

This completes the proof of the second part of Theorem \ref{MDSs}.

$(3).$
This has been done in Corollary \ref{cor2}.

$(4).$
By our assumption $q^2\equiv1\pmod{n}$, every $q$-cyclotomic coset modulo $n$ has size one or two.
Clearly, the congruence  $q(q+1)\equiv q+1\pmod{n}$  implies that the $q$-cyclotomic coset containing $q+1$, denoted by  $C_{q+1}$,
has exactly one element.
Let  $\mathcal{C}$ be a cyclic code of length $n$ over $\mathbb{F}_q$ with defining set
$T=C_0\bigcup C_1\bigcup C_{q+1}$, where $C_0=\{0\}$, $C_1=\{1,q\}$ and $C_{q+1}=\{q+1\}$.
It is easy to see that $\mathcal{C}$ has   dimension $k=n-4$.

We will show that the actual value of $d_H(\mathcal{C})$
is $4$ by using the Hartmann-Tzeng Bound (see  \cite[Theorem 4.5.6]{Huffman}).
Indeed, applying the Hartmann-Tzeng Bound with $A=\{0,1\}$ and $B=\{0,q\}$
(since $\gcd(q,n)=1$), we obtain $d_H(\mathcal{C})\geq3+1=4$.
On the other hand, it follows from the Singleton Bound
(see  \cite[Theorem 2.4.1]{Huffman}) that $d_H(\mathcal{C})\leq n-(n-4)+1=5$.
If the Singleton Bound were met, i.e., $\mathcal{C}$ is an MDS code with parameters $[n,n-4,5]$,
applying    \cite[Corollary 7.4.4]{Huffman}  to $\mathcal{C}$ would give $k=n-4\leq q-1$.
This is a contradiction since we are assuming that $n-4\geq q$.
We conclude that $\mathcal{C}$ is an almost MDS cyclic code over $\mathbb{F}_q$  with parameters
$[n,n-4,4]$. The desired result then follows  immediately from Theorem \ref{add2}.
\qed

\end{document}